\documentclass[12pt]{article}
\usepackage[T1]{fontenc}
\usepackage{authblk}
\usepackage{amssymb}
\usepackage{amsmath,mathtools}
\usepackage{amsfonts}
\usepackage{amsthm}
\usepackage{dsfont}
\usepackage{amssymb}
\usepackage{mathpazo}
\usepackage{fullpage}
\usepackage{enumerate}
\usepackage[colorlinks=true,linkcolor=blue,citecolor=blue,urlcolor=blue]{hyperref}

\theoremstyle{definition}
\newtheorem{question*}{Question}

\newcommand{\abs}[1]{\lvert #1 \rvert}

\newcommand{\norm}[1]{\lVert #1 \rVert}

\newcommand{\tr}{\mathrm{Tr}}

\newcommand{\pc}{P_{\mathcal C}}

\def\ket#1{| #1 \rangle}

\def\kb#1#2{|#1\rangle\!\langle #2 |}

\providecommand{\keywords}[1]
{
  \small	
  \textbf{\textit{Keywords.}} #1
} 

\providecommand{\subjclass}[1]
{
  \small	
  \textbf{\textit{2020 Mathematics Subject Classification.}} #1
}

\theoremstyle{plain}
\newtheorem{theorem}{Theorem}
\newtheorem{lemma}[theorem]{Lemma}
\newtheorem{corollary}[theorem]{Corollary}
\newtheorem{proposition}[theorem]{Proposition}

\theoremstyle{definition}
\newtheorem{definition}[theorem]{Definition}

\newtheorem{example}[theorem]{Example}

\theoremstyle{remark}
\newtheorem*{remark}{Remark}

\begin{document}

\title{Generalized Knill--Laflamme Theorem for Families of Isoclinic Subspaces}

\author{David~W.~Kribs$^1$, Rajesh Pereira$^1$, Mukesh Taank$^1$}

\affil{$^1$Department of Mathematics \& Statistics, University of Guelph, Guelph, ON, Canada N1G 2W1. Email \href{mailto:dkribs@uoguelph.ca}{dkribs@uoguelph.ca},  \href{mailto:pereirar@uoguelph.ca}{pereirar@uoguelph.ca},  \href{mailto:mtaank@uoguelph.ca}{mtaank@uoguelph.ca} }

\date{\today}
\maketitle

\begin{abstract}
Isoclinic subspaces have been studied for over a century. Quantum error correcting codes were recently shown to define a subclass of families of isoclinic subspaces. The Knill--Laflamme Theorem is a seminal result in the theory of quantum error correction, a central topic in quantum information. We show there is a generalized version of the Knill--Laflamme result and conditions that applies to all families of isoclinic subspaces. In the case of quantum stabilizer codes, the expanded conditions are shown to capture logical operators. We apply the general conditions to give a new perspective on a classical subclass of isoclinic subspaces defined by the graphs of anti-commuting unitary operators. We show how the result applies to recently studied mutually unbiased quantum measurements (MUMs), and we give a new construction of such measurements motivated by the approach. 
\end{abstract}



\section{Introduction}

The study of families of isoclinic subspaces naturally grew from the introduction of canonical angles between subspaces in Euclidean geometry, as initiated by Jordan \cite{jordan1875essai} a century and a half ago, which in turn was built on earlier work of Hamilton \cite{hamilton1844theory}. Subsequently, the notion has arisen and been studied in a variety of settings in the context of matrix theory and operator theory; as a selection of examples, we note the (relatively) more recent structural results and formulations of equivalent conditions and applications found in \cite{balla2019equiangular, bjorck1973numerical, garling2011clifford, hoggar1977new, wong1960clifford, wong1977linear, zhang1997quaternions} (see also references therein and forward references). 

On the other hand we have the subject of quantum error correction, which has much more recent origins going back to the early days of modern quantum information theory in the 1990's  \cite{bennett1996ch,gottesman1996d,knill1997knill,knill2000theory,kribs2005quantum,nielsen2002quantum, shor1995pw, steane1996error}. With original (and continuing) motivations coming from the goal to build quantum computers, quantum error correction seems to now touch on all areas of quantum information more generally. 
The Knill--Laflamme Theorem \cite{knill1997knill} is a seminal result in the theory of quantum error correction, and is arguably one of the most important results in all of quantum information theory. It built on important early examples of quantum error correcting codes and played a significant role in giving the subject a firm mathematical and theoretical foundation from which to grow, as a forward reference search on \cite{knill1997knill} will confirm. The theorem gives explicit algebraic conditions to test if a quantum code is correctable for a given set of error operators. It has had many applications in the subject of quantum error correction and beyond; we mention Gottesman's stabilizer formalism \cite{gottesman1996d}, which shows how to construct codes for Pauli error models, as an important such instance. In the context of this discussion, we also note it was recently observed \cite{kribs2019isoclinic} that quantum error correcting codes define a subclass of isoclinic subspace families. 

It is natural to ask, therefore, if there is a generalization of the Knill--Laflamme conditions and theorem that describes all families of isoclinic subspaces? In this paper, we present a positive answer to this question. 
Specifically, we identify and establish the existence of generalized Knill--Laflamme conditions for families of isoclinic subspaces, which in operator form look like this: 
\begin{equation*}
P A_i^* A_j P = \lambda_{ij} U_{ij} P = \lambda_{ij} P U_{ij} ,     
\end{equation*}
where $\{A_i\}_i$ are operators on a Hilbert space, $A_i^*$ are the (Hilbert space) adjoints, $P$ is an orthogonal projection on the space, $\lambda_{ij}$ are complex scalars, and $U_{ij}$ are unitary operators that commute with the projection (or equivalently, partial isometries with $P$ as their initial and final projections). 

The Knill--Laflamme conditions are captured in the special case in which each unitary $U_{ij}$ is the identity operator. We discuss that subclass, giving a brief review and an explanation of how the broadened conditions capture logical operators for stabilizer codes. We revisit a classical family of isoclinic subspaces that are built from graphs of anti-commuting unitary operators \cite{wong1960clifford, wong1961isoclinic, wong1977linear, wong1990normally}, finding a new perspective motivated by the generalized conditions. We show how the result applies to recently studied mutually unbiased quantum measurements (MUMs) \cite{farkas2023mutually, tavakoli2021mutually}, where we apply the conditions to find an alternate proof of their canonical forms, and we give a new construction of such measurements motivated by the approach.     

This paper is organized as follows. Section~\ref{sec:prelims} includes the basic details of isoclinic subspaces and an important equivalent condition we will make use of. Section~\ref{sec:K-L} includes the derivation of the generalized Knill--Laflamme conditions. In Section~\ref{sec:Apps} we consider the three subclasses of isoclinic subspace families noted above in light of these results. We finish in Section~\ref{sec:Conc} with some concluding remarks.

\section{Preliminaries} \label{sec:prelims}

The origins of investigations into families of isoclinic subspaces go back at least a century and a half, where we find the classical notion of canonical angles between pairs of subspaces (sometimes referred to as principal angles) as formulated by Jordan \cite{jordan1875essai}.

Suppose that $\mathcal V$ and $\mathcal W$ are finite dimensional subspaces of a Hilbert space $\mathcal{H}$ and let $k = \min\{\dim(\mathcal V), \dim(\mathcal W)\}.$ Then the {\it canonical angles} $\{ \theta_1, \ldots, \theta_l  \}$ between $\mathcal V$ and $\mathcal W$ are defined as follows: the first canonical angle is the unique number $\theta_1 \in [0,\frac{\pi}{2}]$ such that
\[
\cos(\theta_1) = \max \{\abs{\langle x, y \rangle} : x \in\mathcal V, y\in\mathcal W, \norm{x} = \norm{y} = 1\}.
\]
Let $x_1$ and $y_1$ be unit vectors in $\mathcal V$ and $\mathcal W$ for which the previous maximum is attained. Then we define the second canonical angle as the unique number $\theta_2 \in [0,\frac{\pi}{2}]$ such that
\[
\cos(\theta_2) = \max \{\abs{\langle x, y \rangle} : x \in\mathcal V\cap \{x_1\}^\perp, y \in\mathcal W \cap \{ y_1\}^\perp, \norm{x} = \norm{y} = 1 \}.
\]
For each $k \leq l$, similarly choose unit vectors $x_2, \ldots x_{k-1}$ and $y_2, \ldots y_{k-1}$ in $\mathcal V$ and $\mathcal W$ respectively, in each case where the previous maximum is attained. Then  $\theta_k$ is taken to be the unique number such that
$\cos(\theta_k)$ is equal to the maximum of $\abs{\langle x, y \rangle}$ with unit vectors $x \in\mathcal V \cap \{ x_1,\ldots , x_{k-1} \}^\perp$ and $y \in \mathcal W \cap \{ y_1,\ldots , y_{k-1} \}^\perp$.
More recently (though still over fifty years ago), Bjorck and Golub \cite{bjorck1973numerical} found a computationally-friendly way to find the canonical angles, showing that they can be characterized in terms of the singular values of the product of two matrices that encode their respective subspace.

As the name suggests, subspaces of the same dimension are isoclinic when all their canonical angles are the same. 

\begin{definition}
{\rm
Let $\mathcal V$ and $\mathcal W$ be two $k$-dimensional subspaces of a Hilbert space $\mathcal{H}$, where $1\le k < \dim(\mathcal{H})$. Then $\mathcal V$ and $\mathcal W$ are said to be {\it isoclinic} if all $k$ canonical angles between $\mathcal V$ and $\mathcal W$ are equal. If that angle is $\theta$, then the subspaces are said to be {\it isoclinic at angle $\theta$}.  A family of $k$-dimensional subspaces of a Hilbert space are said to be isoclinic if all pairs of distinct subspaces from the family are pairwise isoclinic.
}
\end{definition}

Any family of mutually orthogonal subspaces are isoclinic at angle $\frac{\pi}{2}$, and this case is often a distinguished special case in certain subclasses (we shall see such instances in the examples below). However, there are many more possibilities, as the following characterization indicates.  

\begin{theorem}\label{isoprop}
Let $\mathcal V$ and $\mathcal W$ be two $k$-dimensional subspaces of a Hilbert space $\mathcal{H}$, with $k \geq 1$. Let $P_{\mathcal{V}}$ and $P_{\mathcal{W}}$ denote the orthogonal projections of $\mathcal H$ onto the subspaces $\mathcal{V}$ and $\mathcal{W}$ respectively. Then the following statements are equivalent: 
\begin{itemize}
\item[$(i)$] $\mathcal V$ and $\mathcal W$ are isoclinic subspaces.
\item[$(ii)$] There exists $\lambda \geq 0$ such that
\begin{equation}\label{isoconds}
P_{\mathcal V}P_{\mathcal W}P_{\mathcal V} = \lambda P_{\mathcal V}   \quad \quad \mathrm{and} \quad \quad P_{\mathcal W}P_{\mathcal V}P_{\mathcal W} = \lambda P_{\mathcal W} .
\end{equation}
Here, $\lambda = \cos(\theta)$ where $\mathcal V$, $\mathcal W$ are isoclinic at angle $\theta$.
\end{itemize}
\end{theorem}

This equivalency was noted without proof in \cite{hoggar1977new} and linked with other conditions around the same time \cite{hoggar1977new,wong1977linear}. This seems to be somewhat typical with the notion of isoclinic subspaces, in that it has arisen in so many places that its descriptions have become almost `folklore' type results. (See \cite{kribs2019isoclinic} for a recent combined proof of various equivalent conditions.)

\section{Knill--Laflamme Type Conditions for Isoclinic Subspaces} \label{sec:K-L}

In this section we will derive conditions for families of isoclinic subspaces that generalize the Knill--Laflamme conditions of quantum error correction. We break up the result into a pair of results that both apply more broadly. 

The Hilbert spaces $\mathcal H$ considered in the following results can be either finite or infinite dimensional, though the specific subclasses we discuss in the next section are all finite dimensional. Given a subspace $\mathcal C$, we will denote the projection of $\mathcal H$ onto the subspace by $\pc$. 

Let us recall some basic properties of partial isometries, which are essential operators in this discussion. An operator $V$ on $\mathcal H$ is a partial isometry if and only if $P = V^*V$, or equivalently $Q = V V^*$, is a (orthogonal) projection. In this case, $P$, respectively $Q$, is called the {\it initial}, respectively {\it final}, projection for $V$, and the following identities are satisfied: 
\[
V= VP = VV^* V = QV = QVP. 
\]
Note that $V$ is a partial isometry if and only if $V^*$ is a partial isometry, with the roles of the initial and final projections reversed. Of course, $V$ is unitary when both of these projections are the identity operator. More general partial isometries are fundamental operators in matrix and operator theory; in particular, we will make use of the polar decomposition of a generic operator $A$ on $\mathcal H$, which gives a partial isometry $V$ such that $A = V |A|$ where $|A| = \sqrt{|A^* A|}$. This partial isometry is not unique, for instance it can be taken to be unitary by extending its action to the whole Hilbert space by defining it unitarily on the orthogonal complement of the initial projection space mapping to a space orthogonal to the final projection space, which does not change the polar decomposition equation. We will make use of some other properties of the polar decomposition operators in the proofs below. 

We first identify generalized Knill--Laflamme conditions that imply the isoclinic subspace projection type equations (Eq.~(\ref{isoconds})) when they are satisfied. 

\begin{theorem}\label{KLcondition} 
Suppose that $\{A_i\}_i$ are operators on a Hilbert space $\mathcal H$ and $\mathcal C\subseteq \mathcal H$ is a subspace such that for all $i,j$,  
\begin{equation}\label{generalKLeqn}
P_{\mathcal C} A_i^* A_j P_{\mathcal C} = \lambda_{ij} U_{ij} P_{\mathcal C} = \lambda_{ij} P_{\mathcal C} U_{ij} ,     
\end{equation}
where $\lambda_{ij}\in \mathbb{C}$ and $U_{ij}$ are unitary operators that commute with $P_{\mathcal C}$. 
Then if $\{ P_i \}_i$ are the range projections of the operators $A_i \pc$, we have 
\begin{equation}\label{projectioneqn}
P_i P_j P_i = | \gamma_{ij} |^2 P_i \quad \quad \forall\, i,j, 
\end{equation}
where 
\begin{equation*}
\gamma_{ij} = \left\{ \begin{array}{cl} \frac{\lambda_{ij}}{\sqrt{\lambda_{ii}\lambda_{jj}}}  & \mathrm{if}\,\, \lambda_{ii} \neq 0 \,\, \mathrm{and} \,\, \lambda_{jj} \neq 0  \\ 0 & \mathrm{otherwise}       \end{array}   \right. 
\end{equation*}
Further, the projections $P_i$ and $P_j$ have the same rank as $\pc$ whenever $\gamma_{ij}\neq 0$. 
\end{theorem}

\begin{proof}
First note that if $\lambda_{ii}=0$ for some $i$, then it follows from Eq.~(\ref{generalKLeqn}) that $A_i \pc = 0 = \pc A_i^*$, and hence that $\lambda_{ij}=0$ for all $j$. Hence for such $i$ and all $j$, Eq.~(\ref{projectioneqn}) is trivially satisfied with $\gamma_{ij} = 0$.  

So for the rest of the proof we assume every $\lambda_{ii}\neq 0$. By Eq.~(\ref{generalKLeqn}), we have that $\lambda_{ii} U_{ii} P_{\mathcal C}$ is a positive operator, and, by dividing both sides of the equation by the positive part of $\lambda_{ii}$, we further observe that this operator must in fact be a positive scalar multiple of $P_{\mathcal C}$. Thus, without loss of generality, we can assume $\lambda_{ii}> 0$ and $U_{ii}=I$ is the identity operator. We can then take the polar decomposition 
\[
A_i \pc = V_i \sqrt{P_{\mathcal C} A_i^* A_i P_{\mathcal C}} = \sqrt{\lambda_{ii}} V_i \pc , 
\]
where $V_i$ is a partial isometry with initial projection $\pc$ and final projection $P_i$. In particular, these operators satisfy the following identities: 
\[
V_i^* V_i = \pc, \quad\quad  V_i V_i^* = P_i, \quad\quad  V_i \pc = V_i, \quad\quad  V_i \pc V_i^* = P_i . 
\]

Thus we have for all such $i,j$, 
\[
\lambda_{ij} U_{ij} P_{\mathcal C} = P_{\mathcal C} A_i^* A_j P_{\mathcal C} = (\sqrt{\lambda_{ii}} \pc V_i^*) (\sqrt{\lambda_{jj}} V_j \pc) = \sqrt{\lambda_{ii}\lambda_{jj}} P_{\mathcal C} V_i^* V_j P_{\mathcal C} , 
\]
from which it follows that 
\[
P_{\mathcal C} V_i^* V_j P_{\mathcal C} = \big( \lambda_{ij}\big(\sqrt{\lambda_{ii}\lambda_{jj}}\big)^{-1} \big) U_{ij} \pc.
\]
If we let this last operator be equal to $B_{ij}$ and put $\gamma_{ij} = \lambda_{ij}(\sqrt{\lambda_{ii}\lambda_{jj}})^{-1}$, then using the fact that $U_{ij}$ is unitary and commutes with $\pc$ we have 
\[
B_{ij} B_{ij}^* = |\gamma_{ij}|^2 (U_{ij} \pc) (U_{ij} \pc)^* = |\gamma_{ij}|^2 \pc. 
\]
On the other hand, we also have 
\begin{align*}
B_{ij} B_{ij}^* = |\gamma_{ij}|^2 \pc &= (P_{\mathcal C} V_i^* V_j P_{\mathcal C}) (P_{\mathcal C} V_j^* V_i P_{\mathcal C}) \\ 
&= \pc V_i^* V_j V_j^* V_i \pc \\ &= \pc V_i^* P_j V_i \pc,  
\end{align*}
where here we have used the algebraic relations between the operators $V_i$, $P_i$ and $\pc$. 
Now multiply the left-side of this equation by $V_i$ and the right-side by $V_i^*$ to obtain the isoclinic equation: 
\[
|\gamma_{ij}|^2 P_i = |\gamma_{ij}|^2 V_i \pc V_i^* = (V_i \pc V_i^* ) P_j  (V_i \pc V_i^*) = P_i P_j P_i.
\]
Finally, note from the argument above that a scalar $\lambda_{ij}$ in Eq.~(\ref{generalKLeqn}) is non-zero if and only if both $\lambda_{ii}$ and $\lambda_{jj}$ are non-zero, and that in such cases the projections $P_i$ and $P_j$ have the same rank as $\pc$ with the partial isometries $V_i$ and $V_j$ intertwining the corresponding subspaces with $\mathcal C$.  
\end{proof}

Note that we cannot conclude in Theorem~\ref{KLcondition} that the family of subspaces are isoclinic only with the satisfaction of Eq.~(\ref{projectioneqn}). While one can show that projections satisfying these equations with non-zero scalars $\gamma_{ij}$ must have the same rank, one can simply have orthogonal subspaces (of any dimension) that trivially satisfy the relations with $\gamma_{ij} =0$. (It could also be added that the orthogonal subspace case is the most mathematically uninteresting in the context of the general isoclinic subspace relations, but it must be allowed for.)

The next result moves us in the converse direction of the previous result. We recall that a partial isometry is nothing more than a unitary operator restricted to a subspace, and hence the following result could equally be phrased with hypotheses given by a family of unitary operators and a subspace, with the projections given by the ranges of the restricted unitaries. 

\begin{lemma}\label{pisom} 
Suppose that $\{V_i\}_i$ are partial isometries on a Hilbert space $\mathcal H$ with common initial space $\mathcal C\subseteq \mathcal H$, and that their final projections $P_i = V_i V_i^*$ satisfy
the following equations for some scalars $\lambda_{ij} \geq 0$: 
\begin{equation}\label{pisomisoeqn}
P_i P_j P_i = \lambda_{ij} P_i \quad \quad \forall\, i,j. 
\end{equation}
Then there are unitary operators $U_{ij}$ that commute with $P_{\mathcal C}$ such that for all $i,j$, 
\begin{equation}\label{generalKLpisomeqn}
P_{\mathcal C} V_i^* V_j P_{\mathcal C} = \gamma_{ij} U_{ij} P_{\mathcal C} = \gamma_{ij} P_{\mathcal C} U_{ij},      
\end{equation} 
where $\gamma_{ij} = \sqrt{\lambda_{ij} }$. 
\end{lemma}

\begin{proof}
First note that $\lambda_{ij} = \lambda_{ji}$ for all $i,j$, since the projections $P_i$, $P_j$ have the same rank (as images of partial isometries with the same initial projection) and Eq.~(\ref{pisomisoeqn}) gives us via the trace, 
\[
\lambda_{ij} \tr(P_i) = \tr(P_i P_j P_i) = \tr (P_i P_j) = \tr(P_jP_iP_j) = \lambda_{ji} \tr(P_j). 
\]
We have $\pc = V_i^* V_i$ for all $i$ and $P_i = V_i V_i^*$. Hence expanding Eq.~(\ref{pisomisoeqn}) yields 
\[
(V_i V_i^*)( V_j V_j^*) (V_i V_i^*) = \lambda_{ij} V_i V_i^*.
\]
We can then multiply both sides of this equation on the left-side by $V_i^*$ and the right-side by $V_i$, and insert $\pc^2$ in the middle of the left-side of the equation (using $V_j\pc=V_j$ and its adjoint equation), to obtain, 
\[
(\pc V_j^* V_i \pc)^* (\pc V_j^* V_i \pc ) = \lambda_{ij} \pc \pc = \lambda_{ij} \pc. 
\]

In the case that $\lambda_{ij}=0 = \lambda_{ji}$ for some pair $i,j$, it follows that $\pc V_j^* V_i \pc  = 0 = 0\pc$ (and the same is true for the adjoint operator with $i,j$ roles reversed), which gives Eq.~(\ref{generalKLpisomeqn}) trivially with $U_{ij} = I$.  

So suppose that $0\leq \lambda_{ij} = \lambda_{ji}\neq 0$ and put $\gamma_{ij} = \sqrt{\lambda_{ij}}$. The above equation thus gives us $\pc = B_{ij}^* B_{ij}$ where we define 
\[
B_{ij} = \gamma_{ij}^{-1} \pc V_j^* V_i \pc. 
\]
Reversing the roles of $i$ and $j$ and using $\gamma_{ij} = \gamma_{ji}$, we also have $\pc = B_{ij} B_{ij}^*$. It follows that $B_{ij}$ is a partial isometry with $\pc$ as both its initial and final projection, from which we have $B_{ij} = B_{ij} \pc = \pc B_{ij}$ and $B_{ij}^* = \pc B_{ij}^* = B_{ij}^* \pc$. 

In particular, this implies that the von Neumann algebra $\mathrm{W}^*(B_{ij})$ generated by $B_{ij}$ (which, in the finite-dimensional case, is the set of polynomials in $B_{ij}$ and $B_{ij}^*$) is contained in the commutant $\{ \pc \}^\prime$ of the projection $\pc$.  
When we examine the polar decomposition 
\[
B_{ij} = U_{ij} \sqrt{B_{ij}^* B_{ij}} = U_{ij} \pc, 
\]
we obtain a unitary $U_{ij} \in \mathrm{W}^*(B_{ij})$ (which is always the case for a unitary in the polar decomposition of an operator \cite{davidson1996c}) and hence $U_{ij}\in\{ \pc \}^\prime$.

We have thus shown that for all $i,j$ with $\lambda_{ij}\neq 0$, there is a unitary $U_{ij}$ that commutes with $\pc$ such that $\pc V_j^* V_i \pc = \sqrt{\lambda_{ij}} U_{ij} \pc$, and this completes the proof. 
\end{proof}

Combining the previous two results in the cases of most relevance to families of isoclinic subspaces gives us the following result. 

\begin{theorem}\label{KLisoclinic} 
Suppose that $\{A_i\}_i$ are operators on a Hilbert space $\mathcal H$ that are scalar multiples of partial isometries with common initial subspace $\mathcal C\subseteq \mathcal H$. Then the range spaces of $A_i \pc$ form an isoclinic family of subspaces if and only if there are scalars $\lambda_{ij}\in \mathbb{C}$ and unitary operators $U_{ij}$ that commute with $P_{\mathcal C}$ such that  
\begin{equation*}
P_{\mathcal C} A_i^* A_j P_{\mathcal C} = \lambda_{ij} U_{ij} P_{\mathcal C} = \lambda_{ij} P_{\mathcal C} U_{ij} \quad \quad \forall\, i,j.      
\end{equation*}
\end{theorem}

\begin{proof} 
By Theorem~\ref{KLcondition}, if Eq.~(\ref{generalKLeqn}) are satisfied, then so are the projection identities given in Eq.~(\ref{projectioneqn}). Whenever the scalar in one of the projection equations is non-zero, the trace argument used in the proof above shows that the corresponding projections have the same rank. The fact that the operators $A_i$ are multiples of partial isometries with the same initial space comes into this direction of the argument just to ensure that, even when the scalars are zero, the range projections have the same rank. Since the range projections have the same rank and satisfy Eq.~(\ref{projectioneqn}), any pair of these range projections satisfy both Eq.~(\ref{isoconds})  and the hypotheses of Theorem \ref{isoprop}.  Therefore, it follows from Theorem \ref{isoprop} that the range spaces form an isoclinic family. 

The other direction of the proof, when the range spaces of the $A_i \pc$ are assumed to be isoclinic, is captured by the special case of Lemma~\ref{pisom} (with Theorem~\ref{isoprop} applied to obtain the projection equations) in which all the projections have the same rank (i.e., even when the scalars in the equation are zero). The scalars that define the $A_i$ as multiples of partial isometries can be absorbed into the scalars $\gamma_{ij}$ of Eq.~(\ref{generalKLpisomeqn}). 
\end{proof} 

\begin{remark}
While this result includes hypotheses on the operators considered, in particular a common support space, we note that it can be applied to every family of isoclinic subspaces in the following way. Given a family $\{ \mathcal V_i \}_{i\geq 1}$ of $k$-dimensional isoclinic subspaces with associated rank-$k$ projections $P_i$, we can pick any $k$-dimensional subspace $\mathcal C = \mathcal V_0$ of $\mathcal H$ with projection $\pc$, and define intertwining partial isometries $V_i$ with final projections $P_i = V_i V_i^*$ and all with the initial projection $\pc = V_i^* V_i$. One can then apply Theorem~\ref{KLisoclinic} with $A_i = P_i$, to obtain Eq.~(\ref{generalKLeqn}) and unitary operators $U_{ij}$. 
As it turns out, in many cases of interest, including for the subclasses considered in the next section, there is a natural `base' or `anchor' subspace for the isoclinic family as in this discussion. 
\end{remark} 

\section{Applications and Examples}\label{sec:Apps}

\subsection{Quantum Error Correcting Codes}

The subclass of isoclinic subspace families defined through quantum error correction come from the actions of sets of error operators $\{ E_i \}$ on a code subspace $\mathcal C \subseteq \mathcal H$ that can be corrected after their actions. Formally, the operators are generally assumed to define a completely positive and trace-preserving (i.e., $\sum_i E_i^* E_i = I$) map given by $\mathcal E(\rho) = \sum_i E_i \rho E_i^*$. The code $\mathcal C$ is then correctable for $\mathcal E$ if there is a completely positive trace-preserving map $\mathcal R$ on $\mathcal H$ such that $(\mathcal R \circ \mathcal E)(\rho) = \rho$ for all operators $\rho$ on $\mathcal H$ that are supported on $\mathcal C$.

The Knill--Laflamme Theorem \cite{knill1997knill} gives testable conditions for a set of error operators to be correctable for a given code, and it is the special case of the isoclinic conditions above with the unitary operators all equal to the identity operator. In other words, $\mathcal C$ is correctable for $\mathcal E$ if and only if there exist scalars $\lambda_{ij}\in \mathbb{C}$ such that for all $i,j$,
\begin{equation}\label{klconditions}
P_{\mathcal C} E_i^* E_j P_{\mathcal C} = \lambda_{ij} P_{\mathcal C}. 
\end{equation}
In this case, the scalars $\Lambda = ( \lambda_{ij} )$ form a positive and trace one (i.e., density) matrix. The recovery operation $\mathcal R$ is then constructed by using the partial isometries obtained in the polar decomposition of the $E_i$ and factoring through the scalar unitary that diagonalizes $\Lambda$ through these equations. See \cite{kribs2005quantum, nielsen2002quantum} for a more comprehensive introduction to quantum error correction. 

It was shown in \cite{kribs2019isoclinic} that for any code subspace $\mathcal C$ and set of (non-degenerate) error operators $\{ E_i \}$, the range subspaces of the restrictions of the $E_i$ to $\mathcal C$ form a family of isoclinic subspaces. As noted above, we can now see this as a distinguished special case of the general isoclinic conditions. There is a partial converse to this result: any pair of isoclinic subspaces can naturally be viewed as arising from a quantum error correcting code with two error operators (essentially the small number of subspaces limits the number of relevant equations and allows for this equivalence). 
More generally, however, the isoclinic subspace conditions of Eq.~(\ref{generalKLeqn}) define a broader class. It is natural to ask though, what, if anything, do these more general equations describe in the context of sets of operators and (code) subspaces that are of relevance in quantum information? To this end, let us consider a motivating class of codes from quantum error correction. 

Let $X,Y,Z$ be the usual (single qubit) Pauli operators on $\mathbb{C}^2$ with orthonormal basis $\{\ket{0}, \ket{1}\}$ \cite{nielsen2002quantum}. Given a positive integer $n\geq1$, we can consider the $n$-tensors of these operators given by products of the operators $X_1 = X \otimes I \otimes I \otimes \ldots$, $Z_2 = I \otimes Z \otimes I \otimes \ldots$, etc, that act on $n$-qubit Hilbert space $(\mathbb{C}^2)^{\otimes n}$, which has orthonormal basis vectors $\ket{i_1\cdots i_n} = \ket{i_1} \otimes \ldots \otimes \ket{i_n}$ given by tensor products of the basis vectors from its single qubit subsystems. The $n$-qubit Pauli group $\mathcal P_n$ is the subgroup of unitary operators generated by the $X_i, Z_i$ and $iI$.  

The stabilizer formalism for quantum error correction invented by Gottesman \cite{gottesman1996d} shows how to build correctable codes for sets of Pauli error operators, and applies the Knill--Laflamme Theorem to give a complete group-theoretic characterization of which error operators are correctable for a given `stabilizer code'. The starting point for the formalism is an Abelian subgroup $\mathcal S$ of $\mathcal P_n$ that does not contain $-I$; for this discussion, let us take $\mathcal S = \langle Z_1, \ldots , Z_{n-k} \rangle$ for some fixed $k \in \{ 1,2, \ldots , n\}$. The stabilizer subspace for $\mathcal S$, which is designated as the code space, is $\mathcal C = \mathcal C(\mathcal S) = \mathrm{span} \{ \ket{\psi} \, : \, S \ket{\psi} = \ket{\psi} \,\, \forall\, S \in \mathcal S\}$. In the example, we get the $k$-qubit subspace $\mathcal C = \mathrm{span} \{ \ket{0}^{\otimes(n-k)} \ket{i_1 \cdots i_k} \, : \, i_j\in\{0,1\} \}$. 

One can check that the normalizer subgroup $\mathcal N(\mathcal S)$ of $\mathcal S$ inside $\mathcal P_n$ coincides with its centralizer $\mathcal Z(\mathcal S)$, as every element of $\mathcal P_n$ either commutes or anti-commutes and $-I \notin \mathcal S$. A main result in the stablizer formalism asserts that a stabilizer code $\mathcal C(\mathcal S)$ is correctable for a set of Pauli error operators $\{E_i\}$ exactly when all the operator products $E_i^* E_j$ satisfy the constraint: $E_i^* E_j \notin \mathcal N(\mathcal S) \setminus \langle \mathcal S, iI \rangle$. In the example, one can see directly that $\mathcal N(\mathcal S) = \mathcal Z(\mathcal S)$ is the group generated by $\langle \mathcal S, iI \rangle$ and the operators $\mathcal L = \langle X_i, Z_i \, : \, n-k+1 \leq i \leq n \rangle$. 

The operators in $\mathcal L$ are called {\it logical operators}, as they can be used to implement logical operations on the code space. (We note that while this is a somewhat canonical choice of logical operators and convenient for our discussion, it is not unique.) They form a multiplicatively closed set (as a group), that clearly belongs to $\mathcal N(\mathcal S) \setminus \langle \mathcal S, iI \rangle$, and hence they are not correctable by the result noted above (in particular they do not satisfy the Knill--Laflamme conditions). However, observe they do satisfy the more general isoclinic equations. Indeed, given $L\in\mathcal L$, we have $L\in \mathcal Z(\mathcal S)$, and so $L$ commutes with all the spectral projections of every element of $\mathcal S$, which in turn implies that $L$ commutes with $\pc$ (this can be seen from the explicit form for $\pc$ noted below). Hence $L\pc = \pc L$ and $\mathcal C$ is a reducing subspace for all $L\in \mathcal L$; i.e., $\mathcal C$ is an invariant subspace for both $L$ and $L^*$. Given any pair $L_1,L_2\in \mathcal L$, we have $L= L_1L_2\in\mathcal L$ and thus $\pc L_1 L_2 \pc = L \pc = \pc L$. So by Theorem~\ref{KLcondition}, it follows that the set of operators $\mathcal L$ (or any subset of them) together with $\mathcal C$ define an isoclinic family of subspaces given by $\{ L\mathcal C : L\in \mathcal L\}$, the ranges of the operators $L$ restricted to $\mathcal C$. 

Now, recall from the definition of $\mathcal C = \mathcal C(\mathcal S)$ that every $S\in \mathcal S$ satisfies $S\pc = \pc = \pc S$. More generally, one can show that an operator $L\in\mathcal P_n$ commutes with $\pc$ if and only if $L\in \mathcal N(\mathcal S)$. Indeed, the eigenvalue 1 eigenspace projection for $S\in \mathcal S$ is $P_S = \frac12 (I+S)$, and $\pc$ is the product of all the $P_S$. If $L\in\mathcal P_n$ does not belong to $\mathcal N(\mathcal S)= \mathcal Z(\mathcal S)$, then it anti-commutes with some $S\in\mathcal S$, and we have $LP_S = \frac12 (I-S)L$, which on the right side of this equation is the eigenvalue -1 eigenspace projection for $S$ multiplied by $L$ on the right. It follows that such an $L$ cannot commute with $\pc$. Thus, it also follows that the set of Pauli group operators that satisfy the general isoclinic equations with a {\it non-trivial} unitary for a given stabilizer code $\mathcal C(\mathcal S)$, are {\it precisely} the operators that belong to $\mathcal N(\mathcal S) \setminus \langle \mathcal S, iI \rangle$. 

Combining this observation with the results of the last section and some basic properties of correctable errors for stabilizer codes, allows us to conclude the following.  

\begin{proposition}\label{qecisoclinic}
Let $\mathcal C = \mathcal C(\mathcal S)$ be an $n$-qubit stabilizer code and let $\{E_i\}_i \subseteq \mathcal P_n$ be any subset of Pauli operators. Then the range subspaces of the operators $E_i \pc$ are isoclinic. Moreover, the products amongst the set $\{ E_i^* E_j \}_{i,j}$ that belong to $\mathcal N(\mathcal S) \setminus \langle \mathcal S, iI \rangle$ are precisely those that satisfy Eq.(\ref{generalKLeqn}) with non-trivial unitary operators. 
\end{proposition} 

\begin{proof} 
In fact we can say precisely how Eq.~(\ref{generalKLeqn}) are satisfied here. Indeed, the argument above shows that any operator product $E_i^* E_j := E$ that belongs to the (logical) operator set $\mathcal N(\mathcal S) \setminus \langle \mathcal S, iI \rangle$, must satisfy these equations with $E$ acting as the non-trivial unitary. (The scalar multiple $\lambda_{ij}$ will be modulus one, as we have assumed the operators belong to $\mathcal P_n$, but we note that in general the error operators will be scalar multiples of such operators, so the error model will be trace-preserving.) The remaining cases correspond to correctable errors. Indeed, If $E$ belongs to $\langle \mathcal S, iI \rangle$, then clearly $E\pc = \pc$ and the condition is trivially satisfied. Further, if $E \notin \mathcal N(\mathcal S)$, then a standard stabilizer formalism argument can be used (using the explicit formula for the projection $\pc$ noted above) to show that $E_i\pc$ and $E_j\pc$ have orthogonal ranges, and hence the isoclinic conditions are satisfied with $\lambda_{ij}=0$.  
\end{proof}

\begin{remark}
This suggests an interesting possibility and line of investigation. We can ask if these isoclinic subspace results for stabilizer codes, and in particular the identification of logical operators with the generalized Knill--Laflamme conditions that have non-trivial commuting unitary operators, extends to arbitrary quantum error correcting codes. 
\end{remark}

\subsection{Isoclinic $n$-Planes in Euclidean $2n$-Space}

Next, we consider a class of examples that form a subclass of isoclinic subspace families with roots dating back over a century. First some background. 

The study of rotations in higher dimensions gained momentum with the discovery of quaternions by Hamilton \cite{hamilton1844theory}, which provided the first systematic approach to understanding rotations in three and four dimensions \cite{hardy1881elements, zhang1997quaternions}. Subsequently Clifford extended this framework to spaces of higher dimensions, introducing what are now called Clifford algebras \cite{garling2011clifford}, which generalized the approach to spaces of higher dimensions. Clifford algebras underpin the algebraic structure of rotations, leading naturally to explorations of planes that exhibit unique rotational properties, including isoclinic rotations as a distinguished case. The concept of isoclinic $n$-planes thus applies the idea of isoclinic subspaces to higher dimensional subspaces within Euclidean vector spaces.

Isoclinic $n$-planes are subspaces of $\mathbb{R}^{n}$ (or, more generally, $\mathbb{C}^{n}$) where rotations affect all vectors in a consistent way: pairs of orthogonal vectors within these planes rotate by the same angle while preserving their orthogonality. They are particularly interesting in even-dimensional spaces and Euclidean spaces, and have special properties that make them useful in higher-dimensional geometry and theoretical physics \cite{wong1960clifford, wong1961isoclinic, wong1977linear, wong1990normally}.
As such, study on even-dimensional space, $\mathbb{R}^{2n}$, has been heavily investigated. An $n$-{\it plane} in $\mathbb{R}^{2n}$ is an $n$-dimensional vector subspace of $\mathbb{R}^{2n}$ accompanied by its induced inner product. 
Two $n$-planes, $\mathcal V, \mathcal W$ are said to be isoclinic with each other if the angle between any non-zero vector in $\mathcal V$ and its projection onto $\mathcal W$ is the same for every non-zero vector in $\mathcal V$. If this angle is given by $\theta$, then $\mathcal V$ and $\mathcal W$ are isoclinic at angle $\theta$ in the sense defined above. 

We can extend the concepts within this class of isoclinic subspaces to utilize the Knill--Laflamme conditions and give a new perspective on the class. We note that \cite{kribs2019isoclinic} contains a similar class of examples for $n=2$ and $2\times 2$ matrices as part of its exposition (though without the Knill--Laflamme viewpoint).  

\begin{example} 
Given a (bounded) operator $A$ on a Hilbert space $\mathcal H$, we can consider its graph, 
\[
\mathcal C_A = \{ (x, Ax ) \, : \, x\in\mathcal H \},
\]
which is a subspace of $\mathcal H \oplus \mathcal H$ that is closed (due to the closed graph theorem in the infinite-dimensional case). We will also consider the subspace of $\mathcal H \oplus \mathcal H$ given by, 
\[
\mathcal C_\infty = \{ (0, x ) \, : \, x\in\mathcal H \}. 
\]

Now suppose that $A$ be a unitary operator on $\mathcal H$. Observe that the projection, which we denote by $P_A$, of $\mathcal H \oplus \mathcal H$ onto $\mathcal C_A$ is given in block matrix form by, 
\[
P_{A} = \frac12 \begin{pmatrix} I & A^* \\ A & I \end{pmatrix}.  
\]
Indeed, an easy way to see this is to note that $P_{A} \, \begin{pmatrix} I & A \end{pmatrix}^t = P_{A}$. Also observe the projection $P_\infty$ onto $\mathcal C_\infty$, is given in block matrix form by, 
\[
P_\infty =  \begin{pmatrix} 0 & 0 \\ 0 & I \end{pmatrix}.  
\]
And note the subspaces $\mathcal C_A$ and $\mathcal C_\infty$ have dimension equal to $\dim \mathcal H$ (for $\mathcal C_A$ using the fact that $A$ is unitary). 
One can verify by direct matrix calculation that for any unitary $A$, we have 
\[ 
P_\infty P_A P_\infty = \frac12 P_\infty \quad \mathrm{and} \quad 
P_A P_\infty P_A = \frac12 P_A . 
\]

Now suppose $B$ is another unitary operator on $\mathcal H$, and calculate the following operator product:  
\begin{align*}
P_A P_B P_A &=  
\frac18 \begin{pmatrix} I  & A^* \\ A  & I \end{pmatrix} \begin{pmatrix} I  & B^* \\ B  & I \end{pmatrix} \begin{pmatrix} I  & A^* \\ A  & I \end{pmatrix} \\ 
&=
\frac18 \begin{pmatrix} 2I + A^*B + B^*A & 2A^* + B^* + A^*BA^*     \\ 2A + B + AB^*A  & 2I + BA^* + AB^* \end{pmatrix} .
\end{align*}
Let us further assume that both $A = A^* = A^{-1}$ and $B = B^* = B^{-1}$ are Hermitian unitary operators and they anti-commute, $AB = -BA$. Then observe that the matrix on the right-side of this equation reduces to a scalar multiple of $P_A$; in fact we have 
\[
P_A P_B P_A = \frac12 P_A . 
\]

It follows that, given any family of Hermitian and pairwise anti-commuting unitary operators acting on a common Hilbert space, the graph subspaces $\mathcal C_A$ of the operators inside the direct sum of the Hilbert space with itself, together with the subspace $\mathcal C_\infty$, form an isoclinic family. We are thus in the scenario of Theorem~\ref{KLisoclinic} and we can investigate what the commuting unitary operators might be in the generalized Knill--Laflamme conditions. Interestingly, for this class of examples, there is a natural base subspace that stands out from the others, namely $\mathcal C_\infty$. 

Hence, we first find partial isometries $V_\infty$ and $V_A$ on $\mathcal H \oplus \mathcal H$ such that $P_\infty = V_\infty V_\infty^*$, $P_A = V_A V_A^*$ and $P_\infty = V_A^* V_A$. One can check that $V_\infty = P_\infty$ and 
\begin{equation}\label{classicalpisom}
V_A = \frac{1}{\sqrt{2}} \begin{pmatrix} 0  & I \\ 0  & A \end{pmatrix}
\end{equation} 
satisfy these identities. Knowing from Theorem~\ref{KLisoclinic} that we will find the commuting unitary form in Eq.~(\ref{generalKLeqn}), we compute to find the following (with $P_\infty$ in place of $\pc$):  
\[
P_\infty V_A^* V_B    P_\infty =  \frac{1}{\sqrt{2}}  V_{A,B}^{(\infty)} P_\infty = \frac{1}{\sqrt{2}}  P_\infty V_{A,B}^{(\infty)} , 
\]
where 
\[
V_{A,B}^{(\infty)} : = \begin{pmatrix} 0  & 0 \\ 0  & U_{A,B}^{(\infty)} \end{pmatrix},  
\]
and where $U_{A,B}^{(\infty)} = \frac{I + A^*B}{\sqrt{2}} = \frac{I + AB}{\sqrt{2}}$ is a unitary operator on $\mathcal H$. 
\end{example}

We can use the Pauli group introduced in the last subsection to produce maximal families of such isoclinic families. To see this, we recall the following result of Kestelman which gives a restriction on the dimension of spaces where there are large sets pairwise anti-commuting invertible matrices. 

\begin{theorem} \label{kes} \cite[Theorem 2]{kestelman1961anticommuting} 
If there exist at least $2q$ pairwise anti-commuting invertible $m \times m$ matrices, then $2^q$ divides $m$.  
\end{theorem}

We also note that if we have an even number $\{ A_j\}_{j=1}^{2q}$ of pairwise anti-commuting invertible matrices, then their product $A_1A_2...A_{2q-1}A_{2q}$ is an invertible matrix that  anti-commutes with every element of $\{ A_j\}_{j=1}^{2q}$.  Hence the maximal number of pairwise anti-commuting invertible matrices is always an odd number.

It follows from Theorem~\ref{kes} that if $m = 2^qp$, where $p$ is odd, then the maximum possible number of pairwise anti-commuting invertible $m \times m$ matrices is $2q+1$.  
An inductive construction of Kestleman shows that this upper bound is always attained with a set of Hermitian unitary matrices. We can use elements of the Pauli groups to give explicit maximal sets of examples.  For the base case where $q=1$, the matrices $X\otimes I_p$, $Y\otimes I_p$ and $Z\otimes I_p$ are three anti-commuting Hermitian unitary $m \times m$ matrices.  For the inductive step, suppose $\{ B_i\}_{i=1}^{2q+1}$ are set of $2q+1$ pairwise anti-commuting $m \times m$ Hermitian unitary matrices, then $\{ X\otimes B_i\}_{i=1}^{2q+1}\bigcup \{Y\otimes I_m,Z\otimes I_m\}$ are a set of $2q+3$ pairwise anti-commuting $2m \times 2m$ Hermitian unitary matrices. 

We can then use the construction from earlier in this section together with these anti-commuting Hermitian unitary matrices to construct pairwise isoclinic subspaces. Let $n$ be an even number with $n=2^kp$ where $p$ is odd. Then we have $2k-1$ anti-commuting $\frac{n}{2} \times \frac{n}{2}$ Hermitian unitary matrices. Hence the above construction applied pairwise to these matrices gives $2k$ pairwise isoclinic $\frac{n}{2}$-dimensional subspaces of $\mathbb{C}^n$.

\subsection{Mutually Unbiased Quantum Measurements}

In this section, we show how the generalized Knill--Laflamme conditions can be applied to a recently introduced notion in quantum information theory that is used in quantum cryptography.  

By a {\it $d$-outcome (quantum) measurement} acting on a Hilbert space $\mathcal H$, we mean a set of orthogonal projections $\{P_a \}_{a=1}^d$ with mutually orthogonal ranges that span all of $\mathcal H$, so $\sum_a P_a = I$. The following definition is from \cite{farkas2023mutually, tavakoli2021mutually}, and was motivated by the desire to find higher dimensional versions of mutually unbiased bases \cite{durt2010mutually}, which are captured in the special case with one-dimensional projections. 

\begin{definition} 
Two $d$-outcome measurements $\{P_a \}_{a=1}^d$ and $\{Q_b \}_{b=1}^d$ on a Hilbert space $\mathcal H$ are called {\it mutually unbiased measurements (MUMs)} if 
\begin{equation}\label{MUMdefn}
P_a = d \, P_a Q_b P_a \quad \quad \mathrm{and} \quad \quad Q_b = d \, Q_b P_a Q_b, 
\end{equation} 
for all $1 \leq a,b \leq d$. 
\end{definition}

By computing traces on both sides of these equations (as we did in an earlier proof), one sees that all of the projections $P_a$, $Q_b$ must have the same rank, which we will denote in this section by the integer $k = \dim (P_a \mathcal H) = \dim (Q_b \mathcal H)$. Since the individual sets of projections have mutually orthogonal ranges, and hence trivially satisfy the isoclinic equations Eq.~(\ref{isoconds}) (with $\lambda =0$), it follows from Theorem~\ref{isoprop} and Eq.~(\ref{MUMdefn}) that the entire set of ranges of the projections $\{P_a, Q_b \}_{a,b}$ forms a family of isoclinic subspaces. 

We can thus apply the results of the previous section to MUMs, and doing so yields the following testable conditions for MUMs when given a family of unitary operators and a subspace, paralleling the Knill--Laflamme conditions for quantum error correction. To keep simple notation, we state the result for `2-element' MUMs (i.e., those with two measurements), but the result extends in an obvious way to MUMs of any size. 

\begin{theorem}\label{KLMUM}
A family of unitary operators $\{ U_i \}_i$ on a Hilbert space $\mathcal H$ and a subspace $\mathcal C$ defines a $d$-outcome 2-element MUM if and only if there is a partition of the partial isometries $\{ U_i \pc \}_i$ into two sets $\{ V_a \}_{a=1}^d$,  $\{ W_b \}_{b=1}^d$ for which the projection sets $\{ V_a V_a^*\}_a$ and $\{ W_b W_b^* \}_b$ define measurements and there are unitary operators $U_{ab}$ such that 
\begin{equation}\label{KLMUMconds}
P_{\mathcal C} V_a^* W_b P_{\mathcal C} = \frac{1}{\sqrt{d}} \, U_{ab} P_{\mathcal C} = \frac{1}{\sqrt{d}} \, P_{\mathcal C} U_{ab} \quad \quad \forall\, a,b.      
\end{equation}
\end{theorem}

\begin{proof} 
The forward direction of the proof follows directly from Lemma~\ref{pisom}, and the backward direction follows from Theorem~\ref{KLcondition}. 
\end{proof} 

As shown in \cite{farkas2023mutually}, any pair of finite-dimensional MUMs can be written in a {\it canonical form}, as two sets of $d$ orthogonal (rank-$k$) projections on $\mathbb{C}^k\otimes \mathbb{C}^d$ given by (up to a change of basis), 
\[
P_a = I_k \otimes \kb{a}{a} \quad\quad \forall\, 1\leq a \leq d, 
\]
where $\{ \ket{a} \}_{a=1}^d$ is the computational basis for $\mathbb{C}^d$, and in the same basis, 
\[
Q_b = \frac1d \sum_{i,j}^d V^b_{ij} \otimes \kb{i}{j} \quad \quad \forall\, 1 \leq b \leq d,
\]
where $V^b_{ij}$ are operators on $\mathbb{C}^k$ that satisfy the following operator relations: 
\begin{equation}
\left\{ 
\begin{array}{rcll} 
V^b_{ii} &=& I_k & \forall\, b,i \\ 
(V^b_{ij})^* &=& V^b_{ji} & \forall\, b,i,j \\
V^b_{i_1 i_2} &=& V^b_{i_1 j} V^b_{j i_2}  & \forall\, b,i_1, i_2, j  \\
\sum_b V^b_{ij} &=& \delta_{ij} \, d \, I_k & \forall\, i, j  .
\end{array}
\right.
\end{equation}

Let us show how the canonical form is related to the characterization of MUMs given in Theorem~\ref{KLMUM}. Firstly, if we have a MUM given by rank-$k$ projections $\{ P_a , Q_b\}_{a,b=1}^d$ in its canonical form, we can pick one of the $P_a$ to identify as $\pc$, let us say $\pc = P_1 = I_k \otimes \kb{1}{1}$. Then we can define $V_a = I_k \otimes \kb{a}{1}$ for $1\leq a \leq d$, and 
\[
W_b = \frac{1}{\sqrt{d}} \, \sum_{j=1}^d V_{j1}^b \otimes \kb{j}{1}  \quad \quad \forall\, 1\leq b \leq d.
\]
It then follows that for all $a,b$, we have $V_a^* V_a = \pc$, $P_a = V_a V_a^*$, $W_b^* W_b = \pc$, $Q_b = W_b W_b^*$, and Eq.~(\ref{KLMUMconds}) are satisfied with $U_{ab} = V_{a1}^b \otimes \kb{1}{1}$. 

On the other hand, we can use the generalized Knill--Laflamme conditions of Theorem~\ref{KLMUM} to give an alternate derivation of the MUM canonical form. We state this as a result. 

\begin{corollary}\label{MUMcanonical} 
Suppose that $\{ V_a\}_{a=1}^d$ and $\{ W_b \}_{b=1}^d$ are partial isometries on $\mathcal H$ with $k$-dimensional mutually orthogonal ranges in the individual sets and satisfying the relations of Eq.~(\ref{KLMUMconds}) for a subspace $\mathcal C$. Then there is a unitary $U: \mathcal H \rightarrow \mathbb{C}^k \otimes \mathbb{C}^d$ such that the family of projections $\{ U V_a V_a^* U^* , U W_b W_b^* U^* \}_{a,b=1}^d$ is a MUM in the canonical form. 
\end{corollary} 

\begin{proof} 
First note that we can assume without loss of generality that $V_1$ is a projection and satisfies $V_1 = V_1 V_1^* = \pc$.

Indeed, this can be accomplished by the `pull-back' technique used earlier, wherein we multiply both sides of Eq.~(\ref{KLMUMconds}) on the left by $V_1$ and on the right by $V_1^*$. Then one can verify the equations will be satisfied by the operators $V_a^\prime = V_a V_1^*$, $W_b^\prime = W_b V_1^*$ and $U_{ab}^\prime = V_1 U_{ab} V_1^*$. 

As $\{P_a: = V_a V_a^* \}_a$ is a family of projections with mutually orthogonal ($k$-dimensional) ranges, we can find a unitary $U: \mathcal H \rightarrow \mathbb{C}^k \otimes \mathbb{C}^d$ such that $UP_aU^* = I_k \otimes \kb{a}{a}$ and $V_{a,U}:= UV_a U^* = I_k \otimes \kb{a}{1}$. 

Now observe that for all $b$, we have $W_{b,U}:= U W_b U^* = (UW_b U^*)(UP_1 U^*)$, and so there are operators $V_{j1}^b$ on $\mathbb{C}^k$ such that $W_{b,U} = \sum_{j=1}^d V_{j1}^b \otimes \kb{j}{1}$. Then, by defining $Q_b:= W_{b,U} W_{b,U}^*$, we leave it to the interested reader to verify that the projections $\{ P_a, Q_b\}_{a,b=1}^d$ define a MUM in canonical form. 
\end{proof} 

We next extend the construction that leads to the classical family of isoclinic subspaces from $n$-planes presented in the previous subsection, and use it to construct a family of MUMs.  

\begin{example} 
Let $d\geq 1$ be a fixed positive integer and let $\omega = e^{\frac{2\pi i}{d}}$ be a primitive $d$-th root of unity. Suppose that $A$ is a unitary operator on a Hilbert space $\mathcal H$ with $A^d=I$. We will focus on finite-dimensional spaces here, but the construction goes through for any Hilbert space. 

Let $k= \dim\mathcal H$ and consider the $k$-dimensional (as $A$ is unitary) subspace of $\mathcal K := \mathcal H^{(d)} \cong \mathcal H \otimes \mathbb{C}^d$ given by 
\[
\mathcal C_A = \big\{ (x, Ax, \ldots , A^{d-1} x ) \, : \, x\in\mathcal H \big\}. 
\]
Generalizing the $d=2$ case above, note that the projection $P_A$ of $\mathcal K$ onto $\mathcal C_A$ is given in $d\times d$ block matrix form as,
\[
P_{A} = \frac1d \begin{pmatrix} I & A^* & (A^*)^2 &  \cdots &  (A^*)^{d-1} \\ A & I & A^*  & \dots&(A^*)^{d-2} \\ A^2 & A & I & \ddots &  \vdots\\  \vdots & \vdots & \ddots & \ddots & A^*   \\  A^{d-1} & A^{d-2} & \cdots & A & I \end{pmatrix}.  
\]

We can similarly carry through this subspace and projection description for each of the $d$ unitary operators $\omega^r A$, for $0 \leq r \leq d-1$, replacing $A$ by $\omega^r A$ in the projection above, and obtain the rank-$k$ projections $\{ P_{(\omega^r A)} \}_{r=0}^{d-1}$ with ranges
\[
P_{(\omega^r A)} \mathcal K = \big\{ (x, (\omega^r A) x,  \ldots , (\omega^r A)^{d-1} x ) \, : \, x\in\mathcal H \big\}. 
\]
Observe that this family of ($d$ in total) projections satisfies the completeness relation 
\[
\sum_{r=0}^{d-1} P_{(\omega^r A)} = I, 
\]
which makes use of the cyclotomic polynomial root equation $1 + (\omega^s) + \ldots + (\omega^s)^{d-1} = 0$ for all $0 \leq s \leq d-1$. This also implies 
the projections have mutually orthogonal ranges as, for all $s$,   
\[
P_{(\omega^s A)} = P_{(\omega^s A)} \, I = \sum_{r=0}^{d-1} P_{(\omega^s A)} P_{(\omega^r A)}   = P_{(\omega^s A)} + \sum_{r\neq s} P_{(\omega^s A)} P_{(\omega^r A)}, 
\]
and so each $P_{(\omega^s A)} P_{(\omega^r A)} =0$. 

Below we will also make use of a generalization of the projection $P_\infty$, which projects onto the subspace $\mathcal C_\infty = \{ (0, \ldots , 0, x ) \, : \, x\in\mathcal H \}$ of $\mathcal H^{(d)}$. In the $d\times d$ block matrix form used above, $P_\infty$ is represented by the matrix with $I$ in the $(d,d)$ entry and $0$'s in all other entries. 
\end{example} 

We now apply Theorem~\ref{KLMUM} to the above class of projections to construct MUMs.  

\begin{corollary}\label{mumconstruction}
Let $d$ be a positive integer and let $\omega$ be a primitive $d$th root of unity. Suppose $A_1,\dots , A_n$ are unitary operators on a Hilbert space $\mathcal H$ such that:  
\[
A_i^d = I \quad \mathrm{and} \quad A_i A_j = \omega A_j A_i \quad \mathrm{for} \quad 1 \leq i < j \leq n. 
\]
Then the sets of projections $\mathcal P_i = \big\{ P_{(\omega^r A_i)}    \big\}_{r=0}^{d-1}$, for $1\leq i \leq n$, form an $n$-element family of $d$-outcome MUMs, with projections of rank equal to $\dim \mathcal H$. 
\end{corollary}

\begin{proof} 
Let us focus for the proof on a pair of unitary operators $A$, $B$ on $\mathcal H$ with $A^d = I = B^d$ and $AB = \omega BA$. To apply Theorem~\ref{KLMUM}, we will use $\mathcal C = \mathcal C_\infty$ as the anchor subspace recall with projection $P_\infty$. We define (for any such $A$) a partial isometry $V_A$ on $\mathcal H^{(d)}$ with adjoint given in $d\times d$ block matrix form by 
\[
V_A^* = \frac{1}{\sqrt{d}} \begin{pmatrix} 0  & 0 & \ldots & 0 \\ \vdots  & \vdots & \vdots & \vdots \\ 0  & 0 & \ldots & 0 \\ I  & A^* & \ldots & (A^*)^{d-1}  \end{pmatrix} .  
\]
Then observe that $P_A = V_A V_A^*$ and $P_\infty = V_A^* V_A$, and $V_A = V_A P_\infty$. 
As we did above, we can replace $A$ with $\omega^rA$, for $0\leq r \leq d-1$, to define $V_{(\omega^r A)}$ accordingly. 

It remains to verify Eq.~(\ref{KLMUMconds}) are satisfied by this family of partial isometries (with the $V_{(\omega^r A)}$, respectively $V_{(\omega^s B)}$, playing the roles of $V_a$'s, respectively $W_b$'s). To this end, observe that for all $0\leq r,s\leq d-1$, we have 
\[
P_\infty V_{(\omega^r A)}^* V_{(\omega^s B)}    P_\infty = V_{(\omega^r A)}^* V_{(\omega^s B)}   =   \frac{1}{\sqrt{d}}  V_{(A,r,B,s)}^{(\infty)} P_\infty = \frac{1}{\sqrt{d}}  P_\infty V_{(A,r,B,s)}^{(\infty)} ,  
\]
where $V_{(A,r,B,s)}^{(\infty)}$ given by 
\[
V_{(A,r,B,s)}^{(\infty)} = \begin{pmatrix} 0  & \ldots & 0 \\ \vdots & \vdots & \vdots \\ 0 & \ldots  & U_{(A,r,B,s)} \end{pmatrix},  
\]
with 
\[
U_{(A,r,B,s)} = \frac{1}{\sqrt{d}} \sum_{j=0}^{d-1} (\omega^{s-r})^j \, (A^*)^j B^j   . 
\]
Thus we complete the proof by showing $U_{(A,r,B,s)}$ is unitary. Begin the calculation as follows, using the facts that $A^* = A^{d-1}$, $B^* = B^{d-1}$ and their anti-commutation relation: 
\begin{align*}
    U_{(A,r,B,s)} U_{(A,r,B,s)}^* &= \frac{1}{d} \sum_{j_1,j_2=0}^{d-1} (\omega^{s-r})^{j_1-j_2} \, (A^*)^{j_1} B^{j_1} (B^*)^{j_2} A^{j_2} \\   
    &= \frac{1}{d} \sum_{j_1,j_2=0}^{d-1} (\omega^{s-r})^{j_1-j_2} \, (A^*)^{j_1} B^{[j_1 + j_2(d-1)]}  A^{j_2} \\ 
    &= I \, + \, \frac{1}{d} \sum_{j_1 \neq j_2}  \omega^{f(r,s,j_1,j_2)} \, (A^{[j_1(d-1) + j_2]}) (B^{[j_1 + j_2(d-1)]}), 
\end{align*}
where $f(r,s,j_1,j_2) = (s-r) (j_1-j_2) - j_2[j_1 + j_2(d-1)]$. 
In the sum, make the substitution $i = j_1 - j_2 \, (\mathrm{mod}\,d)$ for $j_1\neq j_2$, and observe that:
\[
\left\{ 
\begin{array}{lcll}
     \hfill f(r,s,j_1,j_2) & \equiv & (s-r)i - j_2 i & (\mathrm{mod}\,d)  \\
     \hfill j_1(d-1) + j_2 & \equiv & -i & (\mathrm{mod}\,d) \\
     \hfill  j_1 + j_2(d-1) & \equiv & i & (\mathrm{mod}\,d)
\end{array}
\right. . 
\]
Hence, with this substitution, the (non-identity) sum in the last line of the calculation above becomes: 
\[
\frac{1}{d} \sum_{i=1}^{d-1} \sum_{j_2=0}^{d-1}  \omega^{[(s-r)i - j_2 i]} \, A^{-i} B^{i} = 
\frac{1}{d} \sum_{i=1}^{d-1} \omega^{(s-r)i} \Big( \sum_{j_2=0}^{d-1}  ( \omega^{-i})^{j_2} \Big) \, A^{-i} B^{i} = 0, 
\]
again here using the cyclotomic polynomial identity. It now follows that $U_{(A,r,B,s)}$ is unitary. 

Thus we have shown that Eq.~(\ref{KLMUMconds}) are satisfied for this family of partial isometries, and so it follows from Theorem~\ref{KLMUM}, and extending the above argument to the whole family $\{A_i\}_{i=1}^n$, that the projection families $\{\mathcal P_i\}_{i=1}^n$ form a MUM; in particular, we have    
\[
P_{(\omega^r A_i)} P_{(\omega^s A_j)} P_{(\omega^r A_i)} = \frac1d P_{(\omega^r A_i)} , 
\]
for all $0 \leq r,s \leq d-1$ and all $1 \leq i < j \leq n$.  
\end{proof}

\begin{remark}
This MUM construction appears to be new. It is also `coordinate-free', in that it only relies on the algebraic relations of the defining unitary operators and not a particular Hilbert space representation of them. It would be interesting to investigate possible connections between the class of MUMs that are generated by this construction and constructions determined by specific representations (e.g., such as the $n$-qudit construction of Theorem~2.16 in \cite{farkas2023mutually}). 
\end{remark}

\section{Concluding Remarks} \label{sec:Conc}

This work has brought together the classical topic of isoclinic subspaces with the modern topic of quantum error correction, with a generalization of a key result from the latter to the former. This generates a number of potentially new lines of investigation, some as direct outgrowths of the current work and others more speculative. We briefly note a few possibilities here. 

The Knill--Laflamme theorem and conditions have sparked several generalizations and applications even within the subject of quantum error correction itself. As noted above, key results in the stabilizer formalism \cite{gottesman1996d} rely on it, and as we showed, the generalized conditions presented here capture the logical operators for a stabilizer code along with the code's correctable error sets. We wonder if this idea, and in particular using the general conditions to algebraically describe logical operators, extends to more general codes and error models.    

It would be interesting to know if extensions of the Knill--Laflamme theory for other types of quantum error correction could themselves have natural extensions that would generalize the conditions considered here. Specifically, we note the operator algebra generalizations of quantum error correction \cite{beny2007generalization,beny2007quantum,beny2009quantum} and related notions and applications such as complementarity with private codes \cite{crann2016private,jochym2014quantum,kretschmann2008complementarity, kribs2018quantum}. From  matrix theory and (operator) dilation theory perspectives, we further expect that the notions of higher-rank numerical ranges \cite{choi2008geometry,choi2006higher2,choi2005quantum,GLPS,li2011generalized,li2007higher,li2009condition,li2008canonical,martinez2008higher,woerdeman2008higher}, both singular and joint, consideration of which was initially motivated by the Knill--Laflamme conditions, have extensions motivated by the generalized conditions.     

Whenever a generalized framework brings different notions under the same umbrella, new crossover lines of investigation can also be considered. The applications and examples we presented in Section~4 have already provided some indications of this; indeed, this included Pauli groups motivated by quantum error models giving new examples of constructions of isoclinic subspaces, the generalized conditions giving a new perspective on a classical family of isoclinic subspaces, the classical construction motivating a new construction of MUMs, and the generalized conditions giving an alternate construction of the canonical form for MUMs, which we think suggests a deeper (and yet to be explored) connection between the theory of MUMs and quantum error correction.  

We also wonder if the generalized Knill--Laflamme conditions might provide new information or perspectives when applied to other families of isoclinic subspaces beyond what we have considered here. For instance, the isoclinic subspace conditions have been shown to arise in the study of quantum designs \cite{zauner1999quantum}. 

We plan to undertake some of these investigations elsewhere and we invite other interested researchers to do the same.

\vspace{0.1in}

\noindent{\bf Acknowledgements.}

The first author was partially supported by the NSERC Discovery Grant RGPIN-2024-400160. The second author was partially supported by the NSERC Discovery Grant RGPIN-2022-04149.

\bibliographystyle{plain}
\bibliography{refs}

\end{document}